\documentclass[11pt,letterpaper]{amsart} % 11.23 use amsart 

\usepackage{amsmath,amsfonts,amsthm,amssymb,stmaryrd,bm,cite,enumerate}
\theoremstyle{plain}

\numberwithin{equation}{section}

\newtheorem{thm}{Theorem}[section]
\newtheorem{lem}[thm]{Lemma}
\newtheorem{cor}[thm]{Corollary}

\theoremstyle{definition}
\newtheorem{example}{Example}

\allowdisplaybreaks  % introduced 7.15.15

\newcommand{\real}{{\mathbb R}}
\newcommand{\complex}{{\mathbb C}}

\newcommand{\tbullet}{\mathrel{\raise .2ex\hbox{\tiny$\bullet$}}} % 5.8.20 THIS for larger cdot as times
\newcommand{\Ksub}{\mathrm{\scriptscriptstyle{K}}} % 3.1.25
\newcommand{\Jsub}{\mathrm{\scriptscriptstyle{J}}} % 3.1.25

\newcommand{\trace}{\mathrm{tr\,}}
\newcommand{\ityes}{\textit{yes}}      
\newcommand{\itno}{\textit{no}}

\newcommand{\escript}{\mathcal{E}}
\newcommand{\gscript}{\mathcal{G}}
\newcommand{\iscript}{\mathcal{I}}
\newcommand{\jscript}{\mathcal{J}}

\newcommand{\lscript}{\mathcal{L}}
\newcommand{\sscript}{\mathcal{S}}
\newcommand{\alphabar}{\overline{\alpha}}
\newcommand{\betabar}{\overline{\beta}}
\newcommand{\gammabar}{\overline{\gamma}}
\newcommand{\iscriptbar}{\overline{\iscript}}

\newcommand{\Abar}{\overline{A}}
\newcommand{\Bbar}{\overline{B}}
\newcommand{\cbar}{\overline{c}}
\newcommand{\ebar}{\overline{e}}

\newcommand{\onebar}{\overline{1}}
\newcommand{\vbar}{\overline{v}}
\newcommand{\bbar}{\overline{b}}
\newcommand{\ehat}{\widehat{e}}
\newcommand{\etilde}{\tilde{e}}
\newcommand{\utilde}{\tilde{u}}
\newcommand{\vtilde}{\tilde{v}}
\newcommand{\onetilde}{\tilde{1}}

\newcommand{\ab}[1]{\left|#1\right|}
\newcommand{\brac}[1]{\left\{#1\right\}}
\newcommand{\paren}[1]{\left(#1\right)}
\newcommand{\sqbrac}[1]{\left[#1\right]}
\newcommand{\doubleab}[1]{\left|\left|#1\right|\right|}
\newcommand{\sqparen}[1]{{\left[#1\right)}}
\newcommand{\elbows}[1]{{\left\langle#1\right\rangle}}

\begin{document}

\title{GEOMETRIC ALGEBRAS AND FERMION QUANTUM FIELD THEORY}

\author{Stan Gudder}
\address{Department of Mathematics, 
University of Denver, Denver, Colorado 80208}
\email{sgudder@du.edu}
\date{}
\maketitle

\begin{abstract}
Corresponding to a finite dimensional Hilbert space $H$ with $\dim H=n$, we define a geometric algebra $\gscript (H)$ with
$\dim\sqbrac{\gscript (H)}=2^n$. The algebra $\gscript (H)$ is a Hilbert space that contains $H$ as a subspace. We interpret the unit vectors of $H$ as states of individual fermions of the same type and $\gscript (H)$ as a fermion quantum field whose unit vectors represent states of collections of interacting fermions. We discuss creation operators on $\gscript (H)$ and provide their matrix representations. Evolution operators provided by self-adjoint Hamiltonians on $H$ and $\gscript (H)$ are considered. Boson-Fermion quantum fields are constructed. Extensions of operators from $H$ to $\gscript (H)$ are studied. Finally, we present a generalization of our work to infinite dimensional separable Hilbert spaces.
\end{abstract}

\section{Basic Definitions and Preliminary Results}  % Section 1
Unless stated otherwise, all vector spaces are complex and finite dimensional. Although the next three lemmas are known, we include their proofs for completeness.

\begin{lem}    % Lemma 1.1
\label{lem11}
Let $V$ be a vector space with basis $f_1,f_2,\ldots ,f_n$. For $a,b\in V$ with $a=\sum\alpha _if_i$, $b=\sum\beta _if_i$,
$\alpha _i,\beta _i\in\complex$, $i=1,2,\ldots ,n$, define $\elbows{a,b}=\sum\alphabar _i\beta _i$. Then $(V,\elbows{\tbullet ,\tbullet})$
is a complex inner product space.
\end{lem}
\begin{proof}
If $\alpha\in\complex$, then
\begin{align*}
\elbows{a,\alpha b}&=\sum\alphabar _i(\alpha b_i)=\alpha\sum\alphabar _ib_i=\alpha\elbows{a,b}\\
\elbows{a,b}&=\sum\alphabar _i\beta _i=\overline{\sum\alpha _i\betabar _i}=\overline{\elbows{b,a}}
\end{align*}
If $c=\sum\gamma _if_i$, then $a+b=\sum (\alpha _i+\beta _i)f_i$ and
\begin{equation*}
\elbows{c,a+b}=\sum\gammabar _i(\alpha _i+\beta _i)=\sum\gammabar _i\alpha _i+\sum\gammabar _i\beta _i=\elbows{c,a}+\elbows{c,b}
\end{equation*}
We also have
\begin{equation*}
\elbows{a,a}=\sum\alphabar _i\alpha _i=\sum\ab{\alpha _i}^2\ge 0
\end{equation*}
and $\elbows{a,a}=0$ if and only if $\alpha _i=0$, $i=1,2,\ldots ,n$, which is equivalent to $a=0$.
\end{proof}

It follows that the vector space $V$ of Lemma~\ref{lem11} is a Hilbert space with orthonormal basis $f_1,f_2,\ldots ,f_n$. We denote the set of linear operators on $V$ by $\lscript (V)$. If $T\in\lscript (V)$ then $Tf_j=\sum\limits _kT_{kj}f_k$, $T_{kj}\in\complex$ for all $k,j=1,2,\ldots ,n$. We say that the matrix $\sqbrac{T}=\sqbrac{T_{kj}}$ \textit{represents} the operator $T$. Notice that
\begin{equation*}
\elbows{f_k,Tf_j}=\elbows{f_k,\sum _iT_{ij}f_i}=\sum _iT_{ij}\elbows{f_k,f_i}=T_{kj}
\end{equation*}
so we can find $T_{kj}$ explicity.

\begin{lem}    % Lemma 1.2
\label{lem12}
{\rm{(i)}}\enspace If $\sqbrac{T_{kj}}$ represents $T$, then $\alpha\sqbrac{T_{kj}}$, $\alpha\in\complex$, represents $\alpha T$.
{\rm{(ii)}}\enspace If $\sqbrac{T_{kj}}$ represents $T$ and $\sqbrac{S_{kj}}$ represents $S$, then $\sqbrac{T_{kj}+S_{kj}}$ represents $T+S$ and the usual matrix product $\sqbrac{T_{kj}}\sqbrac{S_{kj}}$ represents $TS$.
\end{lem}
\begin{proof}
(i)\enspace This follows from 
\begin{equation*}
(\alpha T)f_j=\alpha Tf_j=\sum _k(\alpha T_{kj})f_k
\end{equation*}
for all  $j=1,2,\ldots ,n$.
(ii)\enspace Since
\begin{equation*}
(T+S)f_j=Tf_j+Sf_j=\sum _kT_{kj}f_k+\sum _kS_{kj}f_k=\sum _k(T_{kj}+S_{kj})f_k
\end{equation*}
we have $\sqbrac{T_{kj}+S_{kj}}$ represents $T+S$. Since
\begin{align*}
(TS)f_j&=T(Sf_j)=T\paren{\sum _kS_{kj}f_k}=\sum _kS_{kj}Tf_k=\sum _kS_{kj}\paren{\sum _iT_{ik}f_i}\\
   &=\sum _{i,k}T_{ik}S_{kj}f_i=\sum _i\paren{\sqbrac{T}\sqbrac{S}}_{ij}f_i
\end{align*}
we have that $\sqbrac{T_{kj}}\sqbrac{S_{kj}}$ represents $TS$.
\end{proof}

If $T\in\lscript (V)$ we define the \textit{adjoint} $T^*\in\lscript (V)$ by $\elbows{T^*a,b}=\elbows{a,Tb}$ for every $a,b\in V$.

\begin{lem}    % Lemma 1.3
\label{lem13}
$S=T^*$ if and only if $\elbows{Sf_j,f_k}=\elbows{f_j,Tf_k}$ for all $j,k=1,2,\ldots ,n$.
\end{lem}
\begin{proof}
If $S=T^*$, then clearly $\elbows{Sf_j,f_k}=\elbows{f_j,Tf_k}$ for all $j,k=1,2,\ldots ,n$. Conversely, suppose
$\elbows{Sf_j,f_k}=\elbows{f_j,Tf_k}$ for all $J,k=1,2,\ldots ,n$. If $a=\sum\alpha _jf_j,b=\sum\beta _kf_k$, then
\begin{align*}
\elbows{Sa,b}&=\elbows{S\sum\alpha _jf_j,\sum\beta _kf_k}=\sum _{j,k}\alphabar _j\beta _k\elbows{Sf_j,f_k}\\
   &=\sum _{j,k}\alphabar _j\beta _k\elbows{f_j,Tf_k}=\elbows{\sum\alpha _jf_j,T\sum\beta _kf_k}\\
   &=\elbows{a,Tb}=\elbows{T^*a,b}
\end{align*}
so $S=T^*$.
\end{proof}

We say $T\in\lscript (V)$ is \textit{self-adjoint} if $T=T^*$. It follows from Lemma~\ref{lem13} that $T$ is self-adjoint if and only if
$\elbows{Tf_j,f_k}=\elbows{f_j,Tf_k}$ for all $j,k=1,2,\ldots ,n$. We denote the set of self-adjoint operators on $V$ by $\lscript _S(V)$.
If $S,T\in\lscript _S(V)$, we write $S\le T$ if $\elbows{a,Sa}\le\elbows{a,Ta}$ for all $a\in V$ and call $T\in\lscript _S(V)$ \textit{positive} if
$T\ge 0$ where $0$ is the zero operator. We call $T\in\lscript _S(V)$ an \textit{effect} if $0\le T\le I$ where $I$ is the identity operator and denote the set of effects by $\escript (V)$. An operator $T\in\lscript _S(V)$ is a \textit{projection} if $T=T^2$. It is well-known that projections are effects and we call projections \textit{sharp effects}. The \textit{trace} of $T\in\lscript (V)$ is
$\trace (T)=\sum\elbows{f_j,Tf_j}$. We call
$\rho\in\lscript _S(V)$ a \textit{state} if $\rho\ge 0$ and $\trace (\rho )=1$. The set of states is denoted by $\sscript (V)$. Finally, an operator $T\in\lscript (V)$ is \textit{unitary} if $TT^*=I$ or equivalently $T^*=T^{-1}$.

We think of a Hilbert space as a mathematical structure that describes a quantum mechanical system \cite{bgl95,blm96,kra83}. In order to understand why this is so, we need to discuss states and effects on $V$. A state $\rho\in\sscript (V)$ corresponds to the initial condition of a quantum system. An effect $A\in\escript (V)$ corresponds to a \ityes--\itno (true-false) measurement or experiment on the quantum system
\cite{hz12,kra83,nc00}. If $A$ results in the outcome \ityes\ when it is measured, we say that $A$ \textit{occurs} and otherwise, it
\textit{does not occur}. It can be shown that $0\le\trace (\rho A)\le 1$ and we call $\trace (\rho A)$ the \textit{probability} that $A$ occurs in the state $\rho$. An \textit{observable} on $V$ is a finite set of effects $A=\brac{A_x\colon x\in\Omega _A}$ where
$\sum\limits _{x\in\Omega _A}A_x=I$ \cite{hz12,nc00}. We call $\Omega _A$ the \textit{outcome set} of $A$ and when $A$ is measured and the resulting outcome $x$ is observed, we say that the effect $A_x$ \textit{occurs}. If $A$ is measured and the system is in state $\rho$, we call $P_\rho ^A(x)=\trace (\rho A_x)$ the \textit{probability distribution} of $A$. Since
\begin{equation*}
\sum _{x\in\Omega _A}P_\rho ^A(x)=\sum _{x\in\Omega _A}\trace (\rho A_x)=\trace\paren{\rho\sum _{x\in\Omega _A}A_x}=\trace (\rho I)
   =\trace (\rho )=1
\end{equation*}
we see that $P_\rho ^A$ is indeed a probability measure. There is a close connection between observables and self-adjoint operators. If
$A=\brac{A_x\colon x\in\Omega _A}$ is an observable and $\brac{\lambda _x\colon x\in\Omega _A}\subseteq\real$ then
$B=\sum _{x\in\Omega _A}\lambda _xA_x$ is a self-adjoint operator. Conversely, if $B\in\lscript (V)$ then by the spectral theorem
\cite{hz12,nc00}, there exist a finite number of sharp effects $A_i$ and real numbers $\lambda _i$, $i=1,2,\ldots ,m$ such that
$\sum A_i=I$ and $B=\sum\lambda _iA_i$. Hence, $A=\brac{A_i\colon i=1,2,\ldots ,m}$ is an observable. There is also a close connection between self-adjoint operators and the dynamics of a quantum system. This is because $T\in\lscript (V)$ is unitary if and only if there exists an $A\in\lscript _S(V)$ such that $T=e^{iA}$ \cite{hz12,nc00}. If $A$ corresponds to the Hamiltonian of a quantum system then the unitary group $U_t=e^{iAt}$, $t=\sqparen{0,\infty}$, describes the dynamics of the system, where $t$ is the time.

A state $\rho$ is \textit{pure} if it is a one-dimensional projection. In this case, there is a unit vector $\psi\in V$ such that
$\rho (a)=\elbows{\psi ,a\psi}$ for every $a\in\escript (V)$ and we write $\rho =\rho _\psi$. Since any state $\rho$ is an affine combination of pure states $\paren{\rho =\sum\lambda _i\rho _i,\lambda _i\ge 0,\sum\lambda _i=1,\rho _i,\hbox{pure}}$ we shall mainly consider only pure states.

\section{Geometric Algebras and Fermion Quantum Fields}  % Section 2 REVISED 7.25
We now show that if $H$ is a complex Hilbert space that describes an individual fermion, then the geometric algebra $\gscript (H)$ over $H$ results in a fermion quantum field theory. Our definition of $\gscript (H)$ differs from the usual algebra in the sense that $\gscript (H)$ is complex while the usual algebra is real \cite{art11,dl03,dfm09,dor02,dys77,hs84,hs03,hs15,mac17}. Let $\dim H=n$ and let $e_1,e_2,\ldots , e_n$ be an orthonormal basis for $H$. The \textit{geometric algebra} $\gscript (H)$ over $H$ is the complex homogeneous, associative, distribution algebra containing $H$ that has the basis consisting of the elements $1\in\complex$
\begin{align*}
\brac{e_i\colon i=1,2,\ldots ,n},&\brac{e_ie_j\colon i,j=1,2,\ldots ,i<j}\\
\left\{e_ie_je_k\colon i,j,k\right.&=\left.1,2,\ldots ,n,i<j<k\right\}\\
&\ \vdots\\
\left\{\ehat _1e_2\cdots e_n, e_1\ehat _2e_3\cdots e_n,\right.&\left.\ldots ,e_1e_2\cdots e_{n-1}\ehat _n\right\}\\
e_1e_2\cdots e_n&=\iscript
\end{align*}
where $e_1e_2\cdots\ehat _i\cdots e_n$ means that $e_i$ is not present. There is one additional axiom for $\gscript (H)$, namely, if
$u=\sum\limits _{j=1}^nc_je_j\in H$, then $uu=\sum\limits _{j=1}^nc_j^2\in\complex$.

If $u=\sum\limits _{j=1}^nc_je_j$, we define $\utilde =\sum\limits _{j=1}^n\cbar _je_j$. It is easy to check that
\begin{equation*}
(\alpha u+\beta v)^\sim =\alphabar\utilde +\betabar\vtilde
\end{equation*}
for all $\alpha ,\beta\in\complex$. If $v=\sum d_je_j$, we obtain
\begin{align*}
uv+vu&=(u+v)(u+v)-uu-vv=\sum _{j=1}^n(c_j+d_j)^2-\sum _{j=1}^nc_j^2-\sum _{j=1}^nd_j^2\\
   &=2\sum _{j=1}^nc_jd_j=2\elbows{\utilde ,v}
\end{align*}
Hence, $\utilde\perp v$ if and only if $uv=-vu$. It also follows that if $j\ne k$, then
\begin{equation*}
\elbows{\etilde _j,e_k}=\elbows{e_j,e_k}=0
\end{equation*}
so $e_je_k=-e_ke_j$. Notice that $uu=\elbows{\utilde,u}$ and if $u=e_1+ie_2$ we have the unusual situation that $u\ne 0$ but $uu=0$. Finally, we have that $uu=\sum\limits _{j=1}^nc_j^2$ for all $u\in H$ if and only if $e_je_j=1$ and $e_je_k=-e_ke_j$ for all $j\ne k$.

An element of the form $e_{i_1}e_{i_2}\cdots e_{i_j}$, $i_r\ne i_s$, is said to have \textit{grade} $j$ and grade $(1)=0$. The set of linear combinations of grade $j$ basis elements is a vector subspace of $\gscript (H)$ called the \textit{grade} $j$ \textit{subspace} and is denoted
$\gscript (H)_j$. By definition, $0$ is considered to be every grade level because we want subspaces. Thus,
$\gscript (H)_0\approx\gscript (H)_n\approx\complex$ and $\gscript (H)_1=H$. We see that
\begin{equation*}
\dim\gscript (H)_j=\binom{n}{j}=\frac{n!}{j!(n-j)!}
\end{equation*}
Hence, $\dim\gscript (H)_0=\dim\gscript (H)_n=1$ and by the binomial formula we have
\begin{equation*}
\dim\gscript (H)=\sum _{j=0}^n\dim\gscript (H)_j=\sum _{j=0}^n\binom{n}{j}=(1+1)^n=2^n
\end{equation*}
For $J_k=\brac{j_1,j_2,\ldots ,j_k}$ with $j_1<j_2<\cdots <j_k,j_i\in\brac{1,2,\ldots ,n}$ we define $e_0=1$
\begin{equation*}
e_{\Jsub_k}=e_{j_1}e_{j_2}\cdots e_{j_k}\in\gscript (H)_k
\end{equation*}
and define $\jscript =\brac{0,J_k\colon k=1,2,\ldots ,n}$. We make $\gscript (H)$ into a Hilbert space by declaring
$\brac{e_{\Jsub}\colon J\in\jscript}$ to be an orthonormal basis for $\gscript (H)$. This follows from the next corollary of Lemma~\ref{lem11}.

\begin{cor}    % Corollary 2.1
\label{cor21}
$\paren{\gscript (H),\elbows{\tbullet ,\tbullet}}$ is a Hilbert space with orthonormal basis\newline
$\brac{e_{\Jsub}\colon J\in\jscript}$ and inner product
$\elbows{a,b}=\sum\limits _{J\in\jscript}\alphabar _{\Jsub}\beta _{\Jsub}$\newline
where $a=\sum\limits _{J\in\jscript}\alpha _{\Jsub}e_{\Jsub}$, $b=\sum\limits _{J\in\jscript}\beta _{\Jsub}e_{\Jsub}$.
\end{cor}

As before, we denote the set of linear operators on $\gscript (H)$ by $\lscript (\gscript\paren{H)}$ and the discussion of Section~1 on operators applies. In particular, if $T\in\lscript\paren{\gscript (H)}$, then $Te_{\Jsub}=\sum\limits _{\Ksub}T_{\Ksub\Jsub}e_{\Ksub}$, $T_{\Ksub\Jsub}\in\complex$ for all $K,J\in\jscript$ and the matrix $\sqbrac{T}=\sqbrac{T_{KJ}}$ represents $T$. Moreover, Lemmas~\ref{lem12} and \ref{lem13} hold. Since
$\gscript (H)$ is an algebra that is also a Hilbert space, we call $\gscript (H)$ a \textit{Hilbert algebra}.

We think of $\gscript (H)$ as a quantum field theory describing a finite number of fermions of the same type. A basis multi-vector $v=e_{i_1}e_{i_2}\cdots e_{i_k}$ represents a state for $k$ fermions of the same type ($k$ electrons or $k$ protons or $k$ neutrons,\ldots ). The actual state is $\rho _v$ but we shall frequently abuse the notation and call any unit vector $a\in\gscript (H)$ a state when we really mean $\rho _a$. The Pauli exclusion principle postulates that two fermions of the same type cannot exist in the same state. This holds in the $\gscript (H)$ framework because if they are in the same state $e_i\in H$, then the resulting state for the pair would be $e_ie_i=1$ which we call the
\textit{vacuum state}. In this sense, the two particles annihilate each other. It is interesting that three particles in the same state
$e_ie_ie_i$ reduces to a single particle in the state $e_i$.

We call the grade $0$ subspace $\gscript (H)_0=\complex$ the \textit{vacuum subspace}, the grade~1 subspace $\gscript (H)_1=H$ the
\textit{one-fermion subspace},\ldots , the grade $j$ subspace $\gscript (H)_j$ the $j$-\textit{fermion subspace}. The reason for this is that
$\gscript (H)_0$ corresponds to the states in which no fermion is present,\ldots , $\gscript (H)_j$ the states in which $j$ fermions are present. In general, we call $e_i$ a \textit{one-fermion state},\ldots $e_{i_1}e_{i_2}\cdots e_{i_j}$ a $j$-\textit{fermion state}. We also have anti-fermions (anti-electrons, anti-protons,\ldots ). We call $\etilde _i=e_1\cdots\ehat _i\cdots e_n$ an \textit{anti-fermion state}, % positions to anti-electrons
\begin{equation*}
(e_ie_j)^\sim = e_1\cdots\ehat _i\cdots\ehat _j\cdots e_n
\end{equation*}
a 2-anti-fermion state, etc. Notice that $\onetilde =\iscript$ and we call $\gscript (H)_n\approx\complex$ the \textit{anti-vacuum} subspace. A fermion and its corresponding anti-fermion annihilate each other to form the anti-vacuum state $\iscript$.

If $a\in\gscript (H)_j$. $\doubleab{a}=1$, we call $\rho _a$ a $j$-\textit{fermion state} and otherwise $\rho _a$ is a
\textit{combination fermion state}. In general, if $a\in\gscript (H)$ with $\doubleab{a}=1$ and $A\in\escript\paren{\gscript (H)}$, the probability that $A$ occurs in the state $\rho _a$ becomes
\begin{align*}
P_{\rho _a}(A)=\trace (\rho _aA)&=\sum _{i\in\jscript}\elbows{e_i,\rho _a(Ae_i)}=\sum _{i\in\jscript}\elbows{e_i,\elbows{a,Ae_i}a}\\
   &=\sum _{i\in\jscript}\elbows{a,Ae_i}\elbows{e_i,a}=\sum _{i\in\jscript}\elbows{Aa,e_i}\elbows{e_i,a}\\
   &=\elbows{Aa,a}=\elbows{a,Aa}
\end{align*}
If $a=\sum _{j\in\jscript}\alpha _je_j$ and $\alpha =(\alpha _j\colon j\in\jscript )$ is the complex vector, we have
\begin{align*}
P_{\rho _a}(A)&=\elbows{\sum _{j\in\jscript}\alpha _je_j,A\sum _{k\in\jscript}\alpha _ke_k}
    =\sum _{j,k\in\jscript}\alphabar _j\alpha _k\elbows{e_j,Ae_k}\\
    &=\elbows{\alpha ,\sqbrac{A_{jk}}\alpha}
\end{align*}

\section{Creation Operators}  % Section 3
If $B\in\gscript (H)$ we define $\Bbar\in\lscript\paren{\gscript (H)}$ by $\Bbar a=Ba$. Notice that
$\overline{(\alpha B)}=\alpha\Bbar ,\overline{(A+B)}=\Abar +\Bbar$ and $\overline{(AB)}=\Abar\ \Bbar$ for all $A,B\in\gscript (H)$. A
particular example is the \textit{creation operator} for a fermion in the state $e_i$ given by
\begin{equation*}
C_{e_i}(a)=\ebar _i(a)=e_ia
\end{equation*}

The following lemma will be useful.

\begin{lem}    % Lemma 3.1
\label{lem31}
$e_1e_2\cdots e_je_1e_2\cdots e_j=1$ if $j=1,4,5,8,9,12,13,\ldots$ and\newline
$e_1e_2\cdots e_je_1e_2\cdots e_j=-1$ if $j=2,3,6,7,10,11,14,15,\ldots$
\end{lem}
\begin{proof}
Clearly $e_1e_1=1$ and we have $e_1e_2e_1e_2=-e_2e_1e_1e_2=-e_2e_2=-1$.
Continuing, we obtain
\begin{equation*}
e_1e_2e_3e_1e_2e_3=e_2e_3e_2e_3=-1
\end{equation*}
by the previous case. For $j=6$ we have
\begin{align*}
e_1e_2e_3e_4&e_5e_6e_1e_2e_3e_4e_5e_6=-e_2e_3e_4e_5e_2e_3e_4e_5\\
   &=e_3e_4e_5e_3e_4e_5=-1
\end{align*}
by the previous case. For $j=7$ we have
\begin{equation*}
e_1e_2\cdots e_7e_1e_2\cdots e_7=e_2e_3\cdots e_7e_2e_3\cdots e_7=-1
\end{equation*}
by the previous case. This pattern continues. For $j=4$, we have
\begin{equation*}
e_1e_2e_3e_4e_1e_2e_3e_4=-e_2e_3e_4e_2e_3e_4=1
\end{equation*}
by the $j=3$ case. For $j=5$, we have
\begin{equation*}
e_1e_2e_3e_4e_5e_1e_2e_3e_4e_5=e_2e_3e_4e_5e_2e_3e_4e_5=1
\end{equation*}
by the previous case. Again the pattern continues.
\end{proof}

\begin{thm}    % Theorem 3.2
\label{thm32}
{\rm{(i)}}\enspace The creation operator $C_{e_i}$ is self-adjoint and unitary.\newline
{\rm{(ii)}}\enspace For $J=\brac{j_1,j_2,\ldots ,j_r}\in\jscript$, the operator $\ebar _{\Jsub}$ is unitary and it is self-adjoint if and only if
$r\in\brac{1,4,5,8,9,12,13,\ldots}$,
\end{thm}
\begin{proof}  % edited _K 3.1.25
(i)\enspace For $J,K\in\jscript$ we have $\elbows{e_{\Ksub},C_{e_i}e_{\Jsub}}=0$ unless $e_{\Ksub}=\pm e_ie_{\Jsub}$ and if
$e_{\Ksub}=\pm e_ie_{\Jsub}$, then
$\elbows{e_{\Ksub},e_ie_{\Jsub}}=\pm 1$. Similarly $\elbows{C_{e_i}e_{\Ksub},e_{\Jsub}}=0$ unless, $e_{\Jsub}=\pm e_ie_{\Ksub}$ and if $e_{\Jsub}=\pm e_ie_{\Ksub}$, then
$\elbows{C_{e_i}e_{\Ksub},e_{\Jsub}}=\pm 1$. Also, $e_{\Ksub}=e_ie_{\Jsub}$ if and only if $e_{\Jsub}=e_ie_{\Ksub}$ and
$e_{\Ksub}=-e_ie_{\Jsub}$ if and only if $e_{\Jsub}=-e_ie_{\Ksub}$. We conclude that
\begin{equation*}
\elbows{e_{\Ksub},C_{e_i}e_{\Jsub}}=\elbows{C_{e_i}e_{\Ksub},e_{\Jsub}}
\end{equation*}
for every $e_{\Jsub},e_{\Ksub}$ so $C_{e_i}=C_{e_i}^*$ and $C_{e_i}$ is self-adjoint. To show that $C_{e_i}$ is unitary, we have
\begin{equation*}
C_{e_i}C_{e_i}e_{\Jsub}=e_ie_ie_{\Jsub}=e_{\Jsub}
\end{equation*}
for every $J\in\jscript$. Hence, $C_{e_i}C_{e_i}^*=C_{e_i}C_{e_i}=I$ so $C_{e_i}$ is unitary.\break
(ii)\enspace The operator $\ebar _{\Jsub}$ is unitary because $\ebar _{\Jsub}=C_{j_1}C_{j_2}\cdots C_{j_r}$ and the product of unitary operators is unitary. We have that $\ebar _{\Jsub}$ is self-adjoint if and only if
\begin{equation*}
C_{j_1}C_{j_2}\cdots C_{j_r}=(C_{j_1}C_{j_2}\cdots C_{j_r})^*=C_{j_r}^*C_{j_{r-1}}^*\cdots C_{j_1}^*=C_{j_r}C_{j_{r-1}}\cdots C_{j_1}
\end{equation*}
This equality holds if and only if
\begin{equation*}
(C_{j_1}C_{j_2}\cdots C_{j_r})^2=C_{j_1}C_{j_2}\cdots C_{j_r}C_{j_r}C_{j_{r-1}}\cdots C_{j_1}=1
\end{equation*}
The result follows from Lemma~\ref{lem31}.
\end{proof}

\begin{example}  % Example 1
Letting $H=\complex ^2$, the algebra $\gscript (H)$ is 4-dimensional with basis
\begin{equation*}
\begin{cases}
\ 1\ & \text{grade } 0 \\
e_1\ e_2& \text{grade } 1\\
\iscript =e_1e_2& \text{grade } 2
\end{cases}
\end{equation*}
The creation operators $C_{e_1},C_{e_2}$ are given by $C_{e_1}(1)=e_1$, $C_{e_1}(e_1)=1$, $C_{e_1}(e_2)=\iscript$,
$C_{e_1}(\iscript )=e_2$ and $C_{e_2}(1)=e_2$, $C_{e_2}(e_1)=-e_1e_2=-\iscript$, $C_{e_2}(e_2)=1$, $C_{e_2}(\iscript )=-e_1$.
The corresponding matrices are
\begin{equation*}
M\sqbrac{C_{e_1}}=\begin{bmatrix}0&1&0&0\\1&0&0&0\\0&0&0&1\\0&0&1&0\end{bmatrix}\quad
M\sqbrac{C_{e_2}}=\begin{bmatrix}0&0&1&0\\0&0&0&-1\\1&0&0&0\\0&-1&0&0\end{bmatrix}
\end{equation*}
It is clear that these matrices are unitary and self-adjoint. The operator $\iscriptbar$ is given by $\iscriptbar (1)=\iscript$, 
$\iscriptbar (e_1)=-e_2$, $\iscriptbar (e_2)=e_1$, $\iscriptbar (\iscript )=-1$. The corresponding matrix is
\begin{equation*}
M\sqbrac{\iscriptbar\,}=\begin{bmatrix}0&0&0&-1\\0&0&1&0\\0&-1&0&0\\1&0&0&0\end{bmatrix}
\end{equation*}
We conclude that $\iscriptbar$ is unitary but not self-adjoint as shown in Theorem~\ref{thm32}(ii).\qedsymbol
\end{example}

\begin{example}  % Example 2
Letting $H=\complex ^3$, the algebra $\gscript (H)$ is 8-dimensional with basis
\begin{equation*}
\begin{cases}
\ 1\ & \text{grade } 0 \\
e_1\ e_2\ e_3& \text{grade } 1\\
e_1e_2\ e_1e_3\ e_2e_3& \text{grade } 2\\
e_1e_2e_3=\iscript& \text{grade } 3
\end{cases}
\end{equation*}
The creation operator $C_{e_1}$ is given by $C_{e_1}(1)=e_1$, $C_{e_1}(e_1)=1$, $C_{e_1}(e_2)=e_1e_2$, $C_{e_1}(e_3)=e_1e_3$,
$C_{e_1}(e_1e_2)=e_2$, $C_{e_1}(e_1e_3)=e_3$, $C_{e_1}(e_2e_3)=\iscript$, $C_{e_1}(\iscript )=e_2e_3$. The corresponding matrix is 
\begin{equation*}
M\sqbrac{C_{e_1}}=\begin{bmatrix}0&1&0&0&0&0&0&0\\1&0&0&0&0&0&0&0\\0&0&0&0&1&0&0&0\\
   0&0&0&0&0&1&0&0\\0&0&1&0&0&0&0&0\\0&0&0&1&0&0&0&0\\0&0&0&0&0&0&0&1\\0&0&0&0&0&0&1&0\\
\end{bmatrix}
\end{equation*}
which is unitary and self-adjoint. Also, $M\sqbrac{C_{e_2}}$, $M\sqbrac{C_{e_3}}$ are similar and are unitary, self-adjoint. The operator
$\overline{e_1e_2}$ satisfies: $\overline{e_1e_2}(1)=e_1e_2$, $\overline{e_1e_2}(e_1)=-e_2$, $\overline{e_1e_2}(e_2)=e_1$,
$\overline{e_1e_2}(e_3)=\iscript$, $\overline{e_1e_2}(e_1e_2)=-1$, $\overline{e_1e_2}(e_1e_3)=-e_2e_3$, $\overline{e_1e_2}(e_2e_3)=e_1e_3$,
$\overline{e_1e_2}(\iscript )=-e_3$. The corresponding matrix is
\begin{equation*}
M\sqbrac{\overline{e_1e_2}\,}=\begin{bmatrix}0&0&0&0&-1&0&0&0\\0&0&1&0&0&0&0&0\\0&-1&0&0&0&0&0&0\\
   0&0&0&0&0&0&0&-1\\1&0&0&0&0&0&0&0\\0&0&0&0&0&0&1&0\\0&0&0&0&0&-1&0&0\\0&0&0&1&0&0&0&0\\
\end{bmatrix}
\end{equation*}
We conclude that $\overline{e_1e_2}$ is unitary but not self-adjoint as shown in Theorem~\ref{thm32}(ii).\qedsymbol
\end{example}

We now consider the eigenvalues and eigenvectors of $C_{e_i}$.

\begin{thm}    % Theorem 3.3
\label{thm33}
The eigenvalues of $C_{e_i}$ are $\pm 1$. The normalized eigenvectors for 1 are $\tfrac{1}{\sqrt{2}}(e_{\Jsub}+e_ie_{\Jsub})$ where
$i\notin J$ and the normalized eigenvectors for $-1$ are $\tfrac{1}{\sqrt{2}}(e_{\Jsub}+e_ie_{\Jsub})$ where $i\notin J$ and the normalized eigenvectors for $-1$ are $\tfrac{1}{\sqrt{2}}(e_{\Jsub}-e_ie_{\Jsub})$ where $i\notin J$. There are $2^{n-1}$ normalized eigenvectors for eigenvalue 1 and $2^{n-1}$ normalized eigenvectors for eigenvalue $-1$.
\end{thm}
\begin{proof}
Since $C_{e_i}$ is self-adjoint and unitary, the eigenvalues of $C_{e_i}$ are real and have absolute value 1. Hence, the eigenvalues
$\lambda$ satisfy $\lambda =\pm 1$. If $i\notin J$ we have
\begin{equation*}
C_{e_i}(e_{\Jsub}+e_ie_{\Jsub})=e_ie_{\Jsub}+e_ie_ie_{\Jsub}=e_ie_{\Jsub}+e_{\Jsub}
\end{equation*}
Hence, $\tfrac{1}{\sqrt{2}}(e_{\Jsub}+e_ie_{\Jsub})$ is a normalized eigenvector for eigenvalue 1 for all $J$ with $i\notin J$. Notice, there are $2^{n-1}$ such eigenvectors. If $i\notin J$ we have
\begin{equation*}
C_{e_i}(e_{\Jsub}-e_ie_{\Jsub})=e_ie_{\Jsub}-e_ie_ie_{\Jsub}=e_ie_{\Jsub}-e_{\Jsub}=-(e_{\Jsub}-e_ie_{\Jsub})
\end{equation*}
Hence, $\tfrac{1}{\sqrt{2}}(e_{\Jsub}-e_ie_{\Jsub})$ is a normalized eigenvector for eigenvalue $-1$ for all $J$ with $i\notin J$. Again, there are $2^{n-1}$ such eigenvectors. Since $\dim\sqbrac{\gscript (H)}=2^n$ we have found all the eigenvectors. of $C_{e_i}$
\end{proof}

Notice that when $i\notin J$ we have
\begin{equation*}
\elbows{e_{\Jsub}+e_ie_{\Jsub},e_{\Jsub}-e_ie_{\Jsub}}
=\elbows{e_{\Jsub},e_{\Jsub}}-\elbows{e_{\Jsub},e_ie_{\Jsub}}+\elbows{e_ie_{\Jsub},e_{\Jsub}}-\elbows{e_ie_{\Jsub},e_ie_{\Jsub}}=0
\end{equation*}
as it should be because eigenvectors for different eigenvalues are orthogonal.

\begin{example}  % Example 3
According to Theorem~\ref{thm33}, if $H=\complex ^2$ the eigenvectors of $C_{e_i}$ in $\gscript (H)$ are as follows. The $J\in\jscript$ for which $1\notin J$ are $J=\brac{0}$ and $J=\brac{2}$. The resulting eigenvectors for eigenvalue $1$ are
\begin{equation*}
\tfrac{1}{\sqrt{2}}(1+e_11)=\tfrac{1}{\sqrt{2}}(1+e_1),\ \tfrac{1}{\sqrt{2}}(e_2+e_1e_2)
\end{equation*}
and the eigenvectors for eigenvalue $-1$ are
\begin{equation*}
\tfrac{1}{\sqrt{2}}(1-e_11)=\tfrac{1}{\sqrt{2}}(1-e_1),\ \tfrac{1}{\sqrt{2}}(e_2-e_1e_2)
\end{equation*}
The corresponding matrix representations for these vectors are
\begin{equation*}
\frac{1}{\sqrt{2}}\begin{bmatrix}1\\1\\0\\0\end{bmatrix},\ \frac{1}{\sqrt{2}}\begin{bmatrix}0\\0\\1\\1\end{bmatrix},\ 
\frac{1}{\sqrt{2}}\begin{bmatrix}1\\-1\\0\\0\end{bmatrix},\ \frac{1}{\sqrt{2}}\begin{bmatrix}0\\0\\1\\-1\end{bmatrix}
\end{equation*}
Applying $M\sqbrac{C_{e_1}}$ to these vector representations verify they are eigenvectors of $C_{e_1}$ for eigenvalues $\pm 1$. We next consider $C_{e_2}$. The $J\in\jscript$ for which $2\notin J$ are $J=\brac{0}$ and $J=\brac{1}$. The resulting eigenvectors for eigenvalue 1 are
\begin{equation*}
\tfrac{1}{\sqrt{2}}(1+e_21)=\tfrac{1}{\sqrt{2}}(1+e_2),\ \tfrac{1}{\sqrt{2}}(e_1+e_2e_1)=\tfrac{1}{\sqrt{2}}(e_1-e_1e_2)
   =\tfrac{1}{\sqrt{2}}(e_1-\iscript )
\end{equation*}
and the eigenvectors for eigenvalue -1 are
\begin{equation*}
\frac{1}{\sqrt{2}}\begin{bmatrix}1\\0\\1\\0\end{bmatrix},\ \frac{1}{\sqrt{2}}\begin{bmatrix}0\\1\\0\\-1\end{bmatrix},\ 
\frac{1}{\sqrt{2}}\begin{bmatrix}1\\0\\-1\\0\end{bmatrix},\ \frac{1}{\sqrt{2}}\begin{bmatrix}0\\1\\0\\1\end{bmatrix}
\end{equation*}
Applying $M\sqbrac{C_{e_2}}$ to these vector representations verify they are eigenvectors of $C_{e_2}$ for eigenvalues $\pm 1$.
\qedsymbol
\end{example}

\begin{example}  % Example 4
We now consider the matrix representations for the eigenvectors of $C_{e_1}$ in $\gscript (H)$ where $H=\complex ^3$. The
$J\in\jscript$ for which $1\notin J$ are $J=\brac{0},\brac{2},\brac{3},\brac{2,3}$ the resulting eigenvectors for eigenvalue 1 are
\begin{equation*}
\tfrac{1}{\sqrt{2}}(1+e_1),\tfrac{1}{\sqrt{2}}(e_2+e_1e_2),\tfrac{1}{\sqrt{2}}(e_3+e_1e_3),\tfrac{1}{\sqrt{2}}(e_2e_3+\iscript )
\end{equation*}
and the eigenvectors for eigenvalue $-1$ are
\begin{equation*}
\tfrac{1}{\sqrt{2}}(1-e_1),\tfrac{1}{\sqrt{2}}(e_2-e_1e_2),\tfrac{1}{\sqrt{2}}(e_3-e_1e_3),\tfrac{1}{\sqrt{2}}(e_2e_3-\iscript )
\end{equation*}
The corresponding matrix representations for these vectors are
\begin{equation*}
\frac{1}{\sqrt{2}}\begin{bmatrix}1\\1\\0\\0\\0\\0\\0\\0\end{bmatrix},\ 
\frac{1}{\sqrt{2}}\begin{bmatrix}0\\0\\1\\0\\1\\0\\0\\0\end{bmatrix},\ 
\frac{1}{\sqrt{2}}\begin{bmatrix}0\\0\\0\\1\\0\\1\\0\\0\end{bmatrix},\ 
\frac{1}{\sqrt{2}}\begin{bmatrix}0\\0\\0\\0\\0\\0\\1\\1\end{bmatrix}
\end{equation*}
\begin{equation*}
\frac{1}{\sqrt{2}}\begin{bmatrix}1\\-1\\0\\0\\0\\0\\0\\0\end{bmatrix},\ 
\frac{1}{\sqrt{2}}\begin{bmatrix}0\\0\\1\\0\\-1\\0\\0\\0\end{bmatrix},\ 
\frac{1}{\sqrt{2}}\begin{bmatrix}0\\0\\0\\1\\0\\-1\\0\\0\end{bmatrix},\ 
\frac{1}{\sqrt{2}}\begin{bmatrix}0\\0\\0\\0\\0\\0\\1\\-1\end{bmatrix}
\end{equation*}
As in Example~3, these vectors form an orthonormal basis for $\gscript (H)$. Applying $M\sqbrac{C_{e_1}}$ to these vector representations verify they are eigenvectors of $C_{e_1}$ for eigenvalues $\pm 1$. Similar results hold for $C_{e_2}$ and $C_{e_3}$.\hfill\qedsymbol
\end{example}

The \textit{anti-commutant} of two operator $S,T$ is 
\begin{equation*}
\brac{S,T}=ST+TS
\end{equation*}

\begin{thm}    % Theorem 3.4
\label{thm34}
{\rm{(i)}}\enspace If $e_1\ne e_2$, then $\brac{C_{e_1},C_{e_2}}=0$.
{\rm{(ii)}}\enspace The eigenvectors $\tfrac{1}{\sqrt{2}}(e_{\Jsub}+e_ie_{\Jsub})$, $\tfrac{1}{\sqrt{2}}(e_{\Jsub}-e_ie_{\Jsub})$, $i\notin J$ form an orthonormal basis for $\gscript (H)$.
\end{thm}
\begin{proof}
(i)\enspace This follows from 
\begin{equation*}
C_{e_1}C_{e_2}a=C_{e_1}e_2a=e_1e_2a=-e_2e_1a=-C_{e_2}C_{e_1}a
\end{equation*}
for all $a\in\gscript (H)$.
(ii)\enspace By Theorem~\ref{thm33}, there are $2^n$ vectors of this form. Since eigenvectors corresponding to distinct eigenvalues of self-adjoint operators are orthogonal the first and second types are mutually orthogonal. Since $i\notin J_1,J_2$, if $J_1\ne J_2$ then the two terms $e_{\Jsub _1},e_ie_{\Jsub _1}$ are different than the two terms $e_{\Jsub _2},e_ie_{\Jsub _2}$. Hence, vectors of the first type are orthogonal to other vectors of the first type and similarly for vectors of the second type. It follows that these vectors form an orthonormal basis for
$\gscript (H)$.
\end{proof}

We now consider the creation operator $C_{e_1}\in\lscript _S\paren{\gscript (\complex ^2)}$ in more detail. The operator $C_{e_2}$ will be similar. Let $\psi _{+1},\psi _{+2}$ be the normalized eigenvectors corresponding to eigenvalue 1 and $\psi _{-1},\psi _{-2}$ be the normalized eigenvectors corresponding to eigenvalue $-1$. Let $P_{\psi _{+1}}$ be the projection onto $\psi _{+1}$. Then
\begin{equation*}
P_{\psi _{+1}}1=\elbows{\psi _{+1},1}\psi _{+1}=\frac{1}{2}\begin{bmatrix}1\\1\\0\\0\end{bmatrix}
\end{equation*}
and similarly
\begin{equation*}
P_{\psi _{+1}}e_1=\frac{1}{2}\begin{bmatrix}1\\1\\0\\0\end{bmatrix},\quad
   P_{\psi _{+1}}e_2=P_{\psi _{+1}}\iscript =\begin{bmatrix}0\\0\\0\\0\end{bmatrix}
\end{equation*}
We conclude that
\begin{equation*}
P_{\psi _{+1}}=\frac{1}{2}\begin{bmatrix}1&1&0&0\\1&1&0&0\\0&0&0&0\\0&0&0&0\end{bmatrix}
\end{equation*}
In a similar way we have
\begin{equation*}
P_{\psi _{+2}}=\frac{1}{2}\begin{bmatrix}0&0&0&0\\0&0&0&0\\0&0&1&1\\0&0&1&1\end{bmatrix}
\end{equation*}
The projection onto the eigenspace for eigenvalue 1 becomes
\begin{equation*}
P_+=\frac{1}{2}\begin{bmatrix}1&1&0&0\\1&1&0&0\\0&0&1&1\\0&0&1&1\end{bmatrix}
\end{equation*}
We consider $P_+$ to be the sharp effect that occurs when a fermion in the state $e_1$ is created.

Now let $P_{\psi _{-1}}$ be the projection onto $\psi _{-1}$. Then
\begin{equation*}
P_{\psi _{-1}}1=\elbows{\psi _{-1},1}\psi _{-1}=\frac{1}{2}\begin{bmatrix}1\\-1\\0\\0\end{bmatrix}
\end{equation*}
and similarly
\begin{equation*}
P_{\psi _{-1}}e_1=\frac{1}{2}\begin{bmatrix}-1\\1\\0\\0\end{bmatrix},\quad
   P_{\psi _{-1}}e_2=P_{\psi _{-1}}\iscript =\begin{bmatrix}0\\0\\0\\0\end{bmatrix}
\end{equation*}
We conclude that 
\begin{equation*}
P_{\psi _{-1}}=\frac{1}{2}\begin{bmatrix}1&-1&0&0\\-1&1&0&0\\0&0&0&0\\0&0&0&0\end{bmatrix}
\end{equation*}
In a similar way we have
\begin{equation*}
P_{\psi _{-2}}=\frac{1}{2}\begin{bmatrix}0&0&0&0\\0&0&0&0\\0&0&1&-1\\0&0&-1&1\end{bmatrix}
\end{equation*}
The projection onto the eigenspace for eigenvalue -1 becomes
\begin{equation*}
P_-=\frac{1}{2}\begin{bmatrix}1&-1&0&0\\-1&1&0&0\\0&0&1&-1\\0&0&-1&1\end{bmatrix}
\end{equation*}
We consider $P_-$ to be the sharp effect that occurs when a fermion in the state $e_1$ is annihilated. As expected we have
$P_++P_-=I$. If the system is initially in the vacuum state 1 them the probability that a fermion in the state $e_1$ is created becomes
\begin{align*}
P^1\hbox{(create }\iscript _1\text{)}=\elbows{1,P_+1}&=\frac{1}{2}\elbows{\begin{bmatrix}1\\0\\0\\0\end{bmatrix},
   \begin{bmatrix}1&1&0&0\\1&1&0&0\\0&0&1&1\\0&0&1&1\end{bmatrix}
    \begin{bmatrix}1\\0\\0\\0\end{bmatrix}}\\
    &=\frac{1}{2}\elbows{\begin{bmatrix}1\\0\\0\\0\end{bmatrix},\begin{bmatrix}1\\1\\0\\0\end{bmatrix}}=\frac{1}{2}
\end{align*}
In a similar way we obtain
\begin{align*}
P^1\hbox{(annihilate }e_1)&=P^{e_1}\hbox{(create }e_1)=P^{e_1}\hbox{(annihilate }e_1)\\
   &=P^{e_2}\hbox{(create }e_1)=P^{e_2}\hbox{(annihilate }e_1)=P^\iscript\hbox{(create }e_1)\\
   &=P^\iscript\hbox{(annihilate }e_1)=1/2
\end{align*}

More generally, suppose the system is in the state
\begin{equation*}
\psi =\frac{\alpha 1+\beta e_1}{\sqrt{\ab{\alpha}^2+\ab{\beta}^2}}
\end{equation*}
Then the probability that a fermion in the state $e_1$ is created becomes
\begin{align*}
P^\psi\hbox{(create }e_1)&=\elbows{\psi ,P_+\psi}=\frac{1}{2(\ab{\alpha}^2+\ab{\beta}^2)}\ 
   \elbows{\begin{bmatrix}\alpha\\\beta\\0\\0\end{bmatrix},\begin{bmatrix}1&1&0&0\\1&1&0&0\\0&0&1&1\\0&0&1&1\end{bmatrix}
   \begin{bmatrix}\alpha\\\beta\\0\\0\end{bmatrix}}\\
   &=\frac{1}{2(\ab{\alpha}^2+\ab{\beta}^2)}\ 
   \elbows{\begin{bmatrix}\alpha\\\beta\\0\\0\end{bmatrix},\begin{bmatrix}\alpha +\beta\\\alpha +\beta\\0\\0\end{bmatrix}}
   %\frac{\alphabar(\alpha +\beta )+\betabar (\alpha +\beta )}{2(\ab{\alpha}^2+\ab{\beta}^2)}
   =\frac{\ab{\alpha +\beta}^2}{2(\ab{\alpha}^2+\ab{\beta}^2)}
\end{align*}
In particular, if $\alpha =\beta =1$ we have $P^\psi\hbox{(create }e_1)=1$

An operator $A\in\lscript _S(\gscript (\complex ^2))$ is an $e_1$-\textit{observable} if it has the form
\begin{equation*}
A=\lambda _1P_{\psi _{+1}}+\lambda _2P_{\psi _{+2}}+\lambda _3P_{\psi _{-1}}+\lambda _4P_{\psi _{-2}}
\end{equation*}
where $\lambda _1,\lambda _2,\lambda _3,\lambda _4\in\real$. In particular $C_{e_1}$ is an $e_1$-observable with
$\lambda _1=\lambda _2=1$, $\lambda _3=\lambda _4=-1$. An $e_1$-observable $A$ has the same eigenvectors as $C_{e_1}$ with corresponding eigenvalues $\lambda _1,\lambda _2,\lambda _3,\lambda _4$. Its general form is
\begin{equation*}
A=\frac{1}{2}\begin{bmatrix}\lambda _1+\lambda _3&\lambda _1-\lambda _3&0&0\\
   \lambda _1-\lambda _3&\lambda _1+\lambda _3&0&0\\0&0&\lambda _2+\lambda _4&\lambda _2-\lambda _4\\
   0&0&\lambda _2-\lambda _4&\lambda _2+\lambda _4\end{bmatrix}
\end{equation*}
If $A$ is measured, its possible outcomes are $\lambda _1,\lambda _2,\lambda _3,\lambda _4$ and when the system is in the state
$\rho$, its probability distribution is
\begin{align*}
P_\rho ^A(\lambda _1)&=\trace (\rho P_{\psi _{+1}}),\quad P_\rho ^A(\lambda _2)=\trace (\rho P_{\psi _{+2}})\\
P_\rho ^A(\lambda _3)&=\trace (\rho P_{\psi _{-1}}),\quad P_\rho ^A(\lambda _4)=\trace (\rho P_{\psi _{-2}})\\
\end{align*}

We now consider the 8-dimensional Hilbert algebra $\gscript (\complex ^3)$. We will establish a pattern that the reader will see carries over to higher dimensions. As before, we consider the creation operator $C_{e_1}\in\lscript _S\paren{\gscript (\complex ^3)}$ and the operators $C_{e_2}$, $C_{e_3}$ will be similar. Let $\psi _{+1},\psi _{+2},\psi _{+3},\psi _{+4}$ be the normalized eigenvectors corresponding to eigenvalue 1 and $\psi _{-1},\psi _{-2},\psi _{-3},\psi _{-4}$ be the normalized eigenvectors corresponding to eigenvalue $-1$. The corresponding projection operators become
\begin{equation*}
P_{\psi _{+1}}=\frac{1}{2}\begin{bmatrix}1&1&0&0&0&0&0&0\\1&1&0&0&0&0&0&0\\0&0&0&0&0&0&0&0\\&&&\vdots\\
0&0&0&0&0&0&0&0\end{bmatrix}
\end{equation*}

\begin{equation*}
P_{\psi _{+2}}=\frac{1}{2}\begin{bmatrix}0&0&0&0&0&0&0&0\\0&0&0&0&0&0&0&0\\0&0&1&0&1&0&0&0\\0&0&0&0&0&0&0&0\\
0&0&1&0&1&0&0&0\\0&0&0&0&0&0&0&0\\0&0&0&0&0&0&0&0\\0&0&0&0&0&0&0&0\end{bmatrix}
\end{equation*}

\begin{equation*}
P_{\psi _{+3}}=\frac{1}{2}\begin{bmatrix}0&0&0&0&0&0&0&0\\0&0&0&0&0&0&0&0\\0&0&0&0&0&0&0&0\\0&0&0&1&0&1&0&0\\
0&0&0&0&0&0&0&0\\0&0&0&1&0&1&0&0\\0&0&0&0&0&0&0&0\\0&0&0&0&0&0&0&0\end{bmatrix}
\end{equation*}

\begin{equation*}
P_{\psi _{+4}}=\frac{1}{2}\begin{bmatrix}0&0&0&0&0&0&0&0\\&&&\vdots&&&&\\0&0&0&0&0&0&0&0\\0&0&0&0&0&0&1&1\\
0&0&0&0&0&0&1&1\end{bmatrix}
\end{equation*}
The projection onto the eigenspace for eigenvalue 1 becomes
\begin{equation*}
P_+=P_{\psi _{+1}}+P_{\psi _{+2}}+P_{\psi _{+3}}+P_{\psi _{+4}}
\end{equation*}
The $-1$ projection operators are
\begin{equation*}
P_{\psi _{-1}}=\frac{1}{2}\begin{bmatrix}1&-1&0&0&0&0&0&0\\-1&1&0&0&0&0&0&0\\0&0&0&0&0&0&0&0\\&&&\vdots\\
0&0&0&0&0&0&0&0\end{bmatrix}
\end{equation*}

\begin{equation*}
P_{\psi _{-2}}=\frac{1}{2}\begin{bmatrix}0&0&0&0&0&0&0&0\\0&0&0&0&0&0&0&0\\0&0&1&0&-1&0&0&0\\0&0&0&0&0&0&0&0\\
0&0&-1&0&1&0&0&0\\0&0&0&0&0&0&0&0\\0&0&0&0&0&0&0&0\\0&0&0&0&0&0&0&0\end{bmatrix}
\end{equation*}

\begin{equation*}
P_{\psi _{-3}}=\frac{1}{2}\begin{bmatrix}0&0&0&0&0&0&0&0\\0&0&0&0&0&0&0&0\\0&0&0&0&0&0&0&0\\0&0&0&1&0&-1&0&0\\
0&0&0&0&0&0&0&0\\0&0&0&-1&0&1&0&0\\0&0&0&0&0&0&0&0\\0&0&0&0&0&0&0&0\end{bmatrix}
\end{equation*}

\begin{equation*}
P_{\psi _{-4}}=\frac{1}{2}\begin{bmatrix}0&0&0&0&0&0&0&0\\&&&\vdots&&&&\\0&0&0&0&0&0&0&0\\0&0&0&0&0&0&1&-1\\
0&0&0&0&0&0&-1&1\end{bmatrix}
\end{equation*}
The projection for eigenvalue $-1$ is $P_-=P_{\psi _{-1}}+P_{\psi _{-2}}+P_{\psi _{-3}}+P_{\psi _{-4}}$.

\section{Boson-Fermion Quantum Fields}  % Section 4
We now briefly consider boson and general boson-fermion quantum field theories. Let $K$ be an $m$-dimensional Hilbert space. The corresponding $r$-\textit{boson Hilbert} space is the Foch space \cite{bgl95,blm96}
\begin{equation*}
H=\complex\oplus K\oplus K^2\oplus\cdots\oplus K^r
\end{equation*}
where $K^i=K\otimes K\otimes\cdots\otimes K (i\hbox{ factors})$ and unit vectors in $K^i$ represent states for $i$ bosons. The
\textit{vacuum space} is $\complex$ and we see that we have states for 0 to $r$ bosons. Letting $k=\dim H$ we have that 
\begin{align*}
k&=1+m+m^2+\cdots +m^r=1+m(1+m+m^2+\cdots +m^{r-1})\\
   &=1+m(k-m^r)=1+mk-m^{r+1}
\end{align*}
Hence, $(m-1)k=m^{r+1}-1$ so $k=\tfrac{m^{r+1}-1}{m-1}$. The simplest nontrivial case is $m=2$, $r=1$, $k=3$. We thus have one boson with two possible basis states $b_1,b_2$ and as we shall see there are three fermions. The next simplest case is $m=2$, $r=2$, $k=7$. In this case we have two bosons with basis states $b_1,b_2$, $b_1\otimes b_1$, $b_1\otimes b_2$, $b_2\otimes b_1$, $b_2\otimes b_2$ and we shall see there are 7 fermions.

Corresponding to $H$ we have the \textit{boson-fermion quantum field} $\gscript (H)$. This quantum field has $r$ bosons and
$k=\dim H=\tfrac{m^{r+1}-1}{m-1}$ fermions. As we have seen, $\gscript (H)$ is a Hilbert algebra with dimension $2^k$. We illustrate this quantum field for the two simple cases mentioned above. In the case $m=2$, $r=1$, $k=3$ we have one boson and three fermions. The bosons and fermions can interact and $\gscript (H)$ has $2^3=8$ basis elements. The Hilbert space $H=\complex\oplus K$ has three bases elements $v,b_1,b_2$ where $v$ is the boson vacuum state and $b_1,b_2$ are boson states. We write the basis states for
$\gscript (H)$ as
\begin{equation*}
\onebar ,\vbar ,\bbar _1,\bbar _2,\vbar\bbar _1,\vbar\bbar _2,\bbar _1\bbar _2,\iscriptbar
\end{equation*}
We interpret $\onebar$ as the fermion vacuum state, $\vbar$ is a fermion that has not interacted with a boson, $\bbar _i$ is a fermion that has interacted with a boson in state $b_i$, $i=1,2$, $\vbar\bbar _i$ represents two fermions the first of which does not interact with a boson and the second interacts with a boson in state $b_i$, $i=1,2$, $\bbar _1\bbar _2$ represents two fermions where the first interacts with a boson in state $b_1$ and the second interacts with a boson in state $b_2$. Finally $\iscriptbar =\vbar\bbar _1\bbar _2$ is the
anti-vacuum state.

Of course, the case $m=2$, $r=2$, $k=7$ is much more complicated because we have two bosons and 7 fermions. In this case
$\gscript (H)$ has $2^7=128$ basis elements. The basis states for $H$ are $v_i,b_1,b_2$, $b_1\otimes b_1$, $b_1\otimes b_2$,
$b_2\otimes b_1$, $b_2\otimes b_2$ and the basis states for $\gscript (H)$ are
\begin{align*}
&\onebar ,\vbar ,\bbar _1,\bbar _2,\overline{b_1\otimes b_1},\overline{b_2\otimes b_1},\overline{b_2\otimes b_2}\\
&\vbar\bbar _1,\vbar\bbar _2,\vbar \overline{(b_1\otimes b_1)},\vbar\overline{(b_1\otimes b_2)},
   \vbar\overline{(b_2\otimes b_1)},\vbar\overline{(b_2\otimes b_2)}\\
&\bbar _1\bbar _2,\bbar _1\overline{(b_1\otimes b_1)},\bbar _1\overline{(b_1\otimes b_2)}, \bbar _1\overline{(b_2\otimes b_1)},
   \bbar _1\overline{(b_2\otimes b_2)}\\
&\bbar _2\overline{(b_1\otimes b_1)},\bbar _2\overline{(b_1\otimes b_2)},\bbar _2\overline{(b_2\otimes b_1)},
   \bbar _2\overline{(b_2\otimes b_2)}\\  % b_2 to \bbar _2
&\overline{(b_1\otimes b_1)}\overline{(b_1\otimes b_2)},\overline{(b_1\otimes b_1)}\overline{(b_2\otimes b_1)},
   \overline{(b_1\otimes b_1)}\overline{(b_2\otimes b_2)}\\
&\overline{(b_1\otimes b_2)}\overline{(b_2\otimes b_1)},\overline{(b_1\otimes b_2)}\overline{(b_2\otimes b_2)},
   \overline{(b_2\otimes b_1)}\overline{(b_2\otimes b_2)}\\
&\vbar\bbar _1\bbar _2,\vbar\bbar _1\overline{(b_1\otimes b_1)},\ldots\\
&\qquad\vdots\\
&\vbar\bbar _1\bbar _2\overline{(b_1\otimes b_1)}\overline{(b_1\otimes b_2)}\overline{(b_2\otimes b_1)},\ldots
    \bbar _1\bbar _2\overline{(b_1\otimes b_1)}\overline{(b_1\otimes b_2)}\overline{(b_2\otimes b_1)}\overline{(b_2\otimes b_2)}
    \iscriptbar
\end{align*}
In this case, we interpret $\overline{b_1\otimes b_1}$ as a fermion that interacts with two bosons both of which in the state
$b_1,\bbar _1\overline {(b_1\otimes b_1)}$ represents two fermions, the first of which interacts with a boson in state $b_1$ and the second interacts with two bosons in the state $b_1\otimes b_1,\vbar\bbar _1\bbar _2$ represents three fermions, the first of which interacts with no boson, the second interacts with a boson in state $b_1$ and the third interacts with a boson in state $b_2$. Higher order cases get exponentially larger. For example, the case $m=3$, $r=2$, $k=13$ with two bosons and 13 fermions gives $\gscript (H)$ with
$2^{13}=8,192$ basis elements.

\section{Evolution Operators}  % Section 5
An operator $U$ on $\gscript (H)$ is unitary if and only if there exists a self-adjoint operator $A$ on $\gscript (H)$ such that $U=e^{i\pi A}$ where the constant $\pi$ is for convenience and is not necessary \cite{hz12,nc00}. We define the \textit{evolution operator}
$U_t=e^{i\pi tA}$, where $t\in\sqparen{0,\infty}$ represents the time and $A$ is called a \textit{Hamiltonian} for the system \cite{bgl95,blm96}. If $\phi$ is a state on $\gscript (H)$, then $U_t(\phi )$ gives the evolution of $\phi$ relative to the Hamiltonian $A$.
If $A$ has spectral representation $A=\sum\lambda _jP_j$, $\lambda _j\in\real$, then 
\begin{equation*}
U_t=e^{i\pi tA}=\sum _je^{i\pi t\lambda _j}P_j=\sum _j\sqbrac{\cos (\pi t\lambda _j)+i\sin (\pi t\lambda _j)}P_j
\end{equation*}
For example, in $\gscript (\complex ^2)$, $C_{e_1}$ is self-adjoint and using $C_{e_1}$ as the Hamiltonian we have
\begin{align*}
U_t=e^{i\pi tC_{e_1}}&=\frac{e^{i\pi t}}{2}  \begin{bmatrix}1&1&0&0\\1&1&0&0\\0&0&1&1\\0&0&1&1\end{bmatrix}
   +\frac{e^{-i\pi t}}{2}\begin{bmatrix}1&-1&0&0\\-1&1&0&0\\0&0&1&-1\\0&0&-1&1\end{bmatrix}\\
&=\begin{bmatrix}\cos (\pi t)&i\sin (\pi t)&0&0\\i\sin (\pi t)&\cos (\pi t)&0&0\\0&0&\cos (\pi t)&i\sin (\pi t)
\\0&0&i\sin (\pi t)&\cos (\pi t)\end{bmatrix}
\end{align*}
In particular, the states $1,e_1,e_2,\iscript$ evolve according to
\begin{align*}
U_t(1)&=\cos (\pi t)1+i\sin (\pi t)e_1\\
U_t (e_1)&=i\sin (\pi t)1+\cos (\pi t)e_1\\
U_t(e_2)&=\cos (\pi t)e_2+i\sin (\pi t)\iscript\\
U_t(\iscript )&=i\sin (\pi t)\iscript _2+\cos (\pi t)\iscript
\end{align*}
Another way to view this is to use the fact that $C_{e_1}$ is unitary so\newline
$C_{e_1}=e^{i\pi A}$ and apply the Hamiltonian $A=\tfrac{-i}{\pi}\ln C_{e_1}$. Since\newline
$\ln (-1)=i\pi$ we have
\begin{equation*}
A=\tfrac{-i}{\pi}\ln C_{e_1}=\tfrac{-i}{\pi}\sqbrac{\ln (1)P_++\ln (-1)P_-}=-\paren{\tfrac{i}{\pi}}i\pi P_-=P_-
\end{equation*}
Hence, letting $U_t=e^{i\pi tA}$ we obtain
\begin{equation*}
U_t=e^{i\pi tP_-}=e^{i\pi t}P_-=\frac{1}{2}
\begin{bmatrix}1+e^{i\pi t}&1-e^{i\pi t}&0&0\\1-e^{i\pi t}&1+e^{i\pi t}&0&0\\0&0&1+e^{i\pi t}&1-e^{i\pi t}
\\0&0&1-e^{i\pi t}&1+e^{i\pi t}\end{bmatrix}
\end{equation*}
In this case, the states $1,e_1,e_2,\iscript$ evolve according to
\begin{align*}
U_1(1)&=\tfrac{1}{2}(1+e^{i\pi t})1+\tfrac{1}{2}(1-e^{i\pi t})e_1\\
U_t(e_1)&=\tfrac{1}{2}(1-e^{i\pi t})1+\tfrac{1}{2}(1+e^{i\pi t})e_1\\
U_t(e_2)&=\tfrac{1}{2}(1+e^{i\pi t})e_2+\tfrac{1}{2}(1-e^{i\pi t})\iscript\\
U_t(\iscript )&=\tfrac{1}{2}(1-e^{i\pi t})e_2+\tfrac{1}{2}(1+e^{i\pi t})\iscript
\end{align*}

Let $A=\lambda _1P_{\psi _{+1}}+\lambda _2P_{\psi _{+2}}+\lambda _3P_{\psi _{-1}}+\lambda _4P_{\psi _{-2}},\lambda _i\in\real$ be a
$C_{e_1}$ observable in $\gscript (\complex ^2)$. The corresponding evolution operator is
\begin{align*}
&U_t=e^{i\pi tA}=e^{i\pi t\lambda _1}P_{\psi _{+1}}+e^{i\pi t\lambda _2}P_{\psi _{+2}}
   +e^{i\pi t\lambda _3}P_{\psi _{-1}}+e^{i\pi t\lambda _4}P_{\psi _{-2}}\\
   &=\frac{e^{i\pi t\lambda _1}}{2}\begin{bmatrix}1&1&0&0\\1&1&0&0\\0&0&0&0\\0&0&0&0\end{bmatrix}
   +\frac{e^{i\pi t\lambda _2}}{2}\begin{bmatrix}0&0&0&0\\0&0&0&0\\0&0&1&1\\0&0&1&1\end{bmatrix}\\
   &\quad +\frac{e^{i\pi t\lambda _3}}{2}\begin{bmatrix}1&-1&0&0\\-1&1&0&0\\0&0&0&0\\0&0&0&0\end{bmatrix}
   +\frac{e^{i\pi t\lambda _4}}{2}\begin{bmatrix}0&0&0&0\\0&0&0&0\\0&0&1&-1\\0&0&-1&1\end{bmatrix}\\
   &=\frac{1}{2}\begin{bmatrix}e^{i\pi t\lambda _1+e^{i\pi t\lambda _3}}&e^{i\pi t\lambda _1}-e^{-i\pi t\lambda _3}&0&0\\
   e^{i\pi t\lambda _1}-e^{i\pi t\lambda _3}&e^{i\pi t\lambda _1}+e^{i\pi t\lambda _3}&0&0\\
   0&0&e^{i\pi t\lambda _2}+e^{i\pi t\lambda _4}&e^{i\pi t\lambda _2}-e^{i\pi t\lambda _4}
   \\0&0&e^{i\pi t\lambda _2}-e^{i\pi t\lambda _4}&e^{i\pi t\lambda _2}+e^{i\pi t\lambda _4}\end{bmatrix}
\end{align*}
The evolution of the states $1,e_1,e_2,\iscript$ are given by
\begin{align*}
U_t(1)&=\tfrac{1}{2}(e^{i\pi t\lambda _1}+e^{i\pi t\lambda _3})1+\tfrac{1}{2}(e^{i\pi t\lambda _1}-e^{i\pi t\lambda _3})e_1\\
U_t(e_1)&=\tfrac{1}{2}(e^{i\pi t\lambda _1}-e^{i\pi t\lambda _3})1+\tfrac{1}{2}(e^{i\pi t\lambda _1}+e^{i\pi t\lambda _3})e_1\\
U_t(e_2)&=\tfrac{1}{2}(e^{i\pi t\lambda _2}+e^{i\pi t\lambda _4})e_2+\tfrac{1}{2}(e^{i\pi t\lambda _2}-e^{i\pi t\lambda _4})\iscript\\
U_t(\iscript )&=\tfrac{1}{2}(e^{i\pi t\lambda _2}-e^{i\pi t\lambda _4})e_2+\tfrac{1}{2}(e^{i\pi t\lambda _2}+e^{i\pi t\lambda _4})\iscript\\
\end{align*}

We next consider the operator $\iscriptbar$ on $\gscript (\complex ^2)$. We know that $\iscriptbar$ is unitary and since
\begin{align*}
\iscriptbar (e_1+ie_2)&=i(e_1+ie_2),\quad\iscript (e_1-ie_2)=-i(e_1-ie_2)\\
\iscript (1+i\iscript )&=-i(1+i\iscript ),\quad \iscript (1-i\iscript )=i(1-i\iscript )
\end{align*}
the eigenvalue $i$ has eigenvectors $\tfrac{1}{\sqrt{2}}(e_1+ie_2)$ and $\tfrac{1}{\sqrt{2}}(1-i\iscript )$ while the eigenvalue $-i$ has eigenvectors $\tfrac{1}{\sqrt{2}}(e_1-ie_2)$ and $\tfrac{1}{\sqrt{2}}(1+i\iscript )$. The projection onto the eigenspace for $i$ is
\begin{equation*}
P_{(i)}=\frac{1}{2}\begin{bmatrix}1&0&0&i\\0&1&-i&0\\0&i&1&0\\-i&0&0&1\end{bmatrix}
\end{equation*}
and the projection onto the eigenspace for $-i$ is
\begin{equation*}
P_{(-i)}=\frac{1}{2}\begin{bmatrix}1&0&0&-i\\0&1&i&0\\0&-i&1&0\\i&0&0&1\end{bmatrix}
\end{equation*}
We now find the Hamiltonian $A$ for the operator $\iscriptbar$. Since $\ln (i)=\tfrac{\pi}{2}i$ and $\ln (-i)=-\tfrac{\pi}{2}i$ and
$\iscriptbar =e^{i\pi A}$ we conclude that
\begin{equation*}
A=\tfrac{-i}{\pi}\ln (\iscriptbar )=\tfrac{-1}{\pi}\sqbrac{\ln (i)P_{(i)}+\ln (-i)P_{(-i)}}=\tfrac{1}{2}P_{(i)}-\tfrac{1}{2}P_{(-i)}
\end{equation*}
The dynamics for $\iscriptbar$ becomes
\begin{align*}
U_t&=e^{i\pi tA}=e^{i\tfrac{\pi}{2}t}P_{(i)}+e^{-i\tfrac{\pi}{2}t}P_{(-i)}
   =(\cos\tfrac{\pi}{2}t+i\sin\tfrac{\pi}{2}t)P_{(i)}+(\cos\tfrac{\pi}{2}t-i\sin\tfrac{\pi}{2}t)P_{(-i)}\\
  &=(\cos\tfrac{\pi}{2}t)I+i\sin\tfrac{\pi}{2}t\sqbrac{P_{(i)}-P_{(-i)}}=(\cos\tfrac{\pi}{2}t)I+\sin\tfrac{\pi}{2}t
  \begin{bmatrix}0&0&0&-1\\0&0&1&0\\0&-1&0&0\\1&0&0&0\end{bmatrix}\\
  &=\begin{bmatrix}\cos\tfrac{\pi}{2}t&0&0&-\sin\tfrac{\pi}{2}t\\0&\cos\tfrac{\pi}{2}t&\sin\tfrac{\pi}{2}t&0\\
  0&-\sin\tfrac{\pi}{2}t&\cos\tfrac{\pi}{2}t&0\\\sin\tfrac{\pi}{2}t&0&0&\cos\tfrac{\pi}{2}t\end{bmatrix}
\end{align*}
The evolution for the states $1,e_1,e_2,\iscript$ are given by
\begin{align*}
U_t(1)&=(\cos\tfrac{\pi}{2}t)1+(\sin\tfrac{\pi}{2}t)\iscript\\
U_t(e_1)&=(\cos\tfrac{\pi}{2}t)1-(\sin\tfrac{\pi}{2}t)e_2\\
U_t(e_2)&=(\sin\tfrac{\pi}{2}t)e_1+(\cos\tfrac{\pi}{2}t)e_2\\
U_t(\iscript )&=-(\sin\tfrac{\pi}{2}t)e_1+(\cos\tfrac{\pi}{2}t)\iscript
\end{align*}
We would like to point out the similarity between the operator $\iscriptbar$ on $\gscript (\complex ^2)$ and the operator
$\overline{e_1e_2}$ on $\gscript (\complex ^3)$. The eigenvalues of $\overline{e_1e_2}$ are $i$ and $-i$ and the eigenvectors for $i$ are $\tfrac{1}{\sqrt{2}}(e_1+ie_2)$, $\tfrac{1}{\sqrt{2}}(e_3-i\iscript )$, $\tfrac{1}{\sqrt{2}}(1-ie_1e_2)$, $\tfrac{1}{\sqrt{2}}e_3(e_1+ie_2)$ and the eigenvectors for $-i$ are $\tfrac{1}{\sqrt{2}}(e_1-ie_2)$, $\tfrac{1}{\sqrt{2}}(e_3+i\iscript )$, $\tfrac{1}{\sqrt{2}}(1+ie_1e_2)$,
$\tfrac{1}{\sqrt{2}}e_3(e_1-ie_2)$. The dynamics for $\overline{e_1e_2}$ are simpler but more complicated than that of $\iscript$.

\section{Extension Operators}  % Section 6
We now discuss extensions of operators from $H$ to $\gscript (H)$. Let $\dim H=n$ and $B\in\lscript (H)$. If
$\alpha =(\alpha _1,\alpha _2,\ldots ,\alpha _n)$, $\alpha _i\in\complex$ we define
\begin{align*}
B^\alpha (1)&=\alpha _11\\
B^\alpha (e_i)&=B(e_i), i=1,2,\ldots ,n\\
B^\alpha (e_1e_j)&=\alpha _2B(e_i)e_j+\alpha _2e_iB(e_j)\hbox{ for }i, j=1,2,\ldots ,n\hbox{ with }i<j\\
B^\alpha (e_ie_je_k)&=\alpha _3\sqbrac{B(e_i)e_je_k+e_iB(e_j)e_k+e_ie_jB(e_k)}\\
&\qquad\hbox{for }i,j,k=1,2,\ldots ,n\hbox{ with }i<j<k\\
&\quad\vdots\\
B^\alpha (e_{i_1}e_{i_2}\cdots e_{i_k})
&=\alpha _k\sqbrac{B(e_{i_1})e_{i_2}\cdots e_{i_k}+e_{i_1}B(e_{i_2})\cdots e_{i_k}+\cdots +e_{i_1}e_{i_2}\cdots B(e_{i_k})}\\
&\qquad\hbox{for }i_1,i_2,\ldots ,i_k=1,2,\ldots ,n\hbox{ with }i_1<i_2<\cdots <i_k\\
&\quad\vdots\\
B^\alpha (\iscript )&=\alpha _n\sqbrac{B(e_1)e_2\cdots e_n+\cdots +e_1e_2\cdots B(e_n)}
\end{align*}
and extend $B^\alpha$ to $\gscript (H)$ by linearity. Then $B^\alpha\in\lscript\paren{\gscript (H)}$ and we call $B^\alpha$ the
$\alpha$-\textit{extension} of $B$. For example, if $\alpha _i=0$ for $i=1,2,\ldots ,n$, we call $B^\alpha$ the \textit{trivial extension} of $B$ and if $\alpha _i=\tfrac{1}{i}$ for $i=1,2,\ldots ,n$, we call $B^\alpha$ the \textit{simple extension} of $B$. Notice that
$(\beta B)^\alpha =\beta B^\alpha$ for any $\beta\in\complex$ and $(A+B)^\alpha =A^\alpha +B^\alpha$ for any $A,B\in\lscript (H)$. However, $(AB)^\alpha\ne A^\alpha B^\alpha$ in general. For example, let $P\in\lscript (H)$ be the projection onto $e_1$. Then
$P^\alpha (e_1e_2)=\alpha _2e_1e_2$ and
$P^\alpha P^\alpha (e_1e_2)=\alpha _2^2e_1e_2\ne P^\alpha (e_1e_2)=(PP)^\alpha (e_1e_2)$. This also shows that if $P$ is a projection, then $P^\alpha$ need not be a projection. Letting $I\in\lscript (H)$ be the identity operator, we have $I^\alpha (e_{i_1}e_{i_2}\cdots e_{i_j})=j\alpha _je_{i_1}e_{i_2}\cdots e_{i_j}$. Hence, $I^\alpha\in\lscript (\gscript (H))$ is the identity operator if and only if
$\alpha _j=\tfrac{1}{j}$ which is equivalent to $I^\alpha$ being a simple extension of $I$.

\begin{thm}    % Theorem 6.1
\label{thm61}
If $B\in\lscript _S(H)$ and $\alpha _i\in\real$, then $B^\alpha\in\lscript _S(\gscript (H))$.
\end{thm}
\begin{proof}
Let $B(e_i)=\sum\limits _{j=1}^nB_{ij}e_j$, $i=1,2,\ldots ,n$. Since $B$ is self-adjoint, we have $B_{ij}=\Bbar_{ji}$. Now
$\elbows{e_{i_1}\cdots e_{i_k},e_re_{j_1}\cdots e_{j_s}}\ne 0$ if and only if $e_{i_1}\cdots e_{i_k}=\pm e_re_{j_1}\cdots e_{j_s}$ and in both of these cases we have
\begin{equation}                % equation (6.1)
\label{eq61}
\elbows{e_{e_1}\cdots e_{i_k},e_re_{j_1}\cdots e_{j_s}}=\elbows{e_re_{i_1}\cdots e_{i_k},e_{j_1}\cdots e_{j_s}}
\end{equation}
We then obtain by \eqref{eq61} that
\begin{align*}
&\elbows{e_{i_1}e_{i_2}\cdots e_{i_r},B^\alpha e_{j_1}e_{j_2}\cdots e_{j_s}}=\alpha _s
\left[\elbows{e_{i_1}e_{i_2}\cdots e_{i_k},B(e_{j_1})e_{j_2}\cdots e_{j_s}}\right.\\
&\qquad\left.+\cdots +\elbows{e_{i_1}e_{i_2}\cdots e_{i_k},e_{j_1}\cdots e_{j_{s-1}}B(e_{j_s})}\right]\\
&=\alpha _s\left[\elbows{e_{i_1}e_{i_2}\cdots e_{i_k},\sum _tB_{j_1t}(e_t)e_{j_2}\cdots e_{j_s}}\right.\\
&\qquad\left.+\cdots +\elbows{e_{i_1}e_{i_2}\cdots e_{i_k},e_{j_1}\cdots e_{j_{s-1}}\sum _tB_{j_st}(e_t)}\right]\\
&=\alpha _s
\left[\sum _tB_{j_1t}\elbows{e_{i_1}e_{i_2}\cdots e_{i_k},e_te_{j_2}\cdots e_{j_s}}+\cdots +\right.\\
&\qquad\left.\sum _tB_{j_st}\elbows{e_{i_1}e_{i_2}\cdots e_{i_k},e_{j_1}\cdots e_{j_{s-1}}e_t}\right]\\
&=\alpha _s\left[\elbows{\sum _tB_{j_1t}(e_t)e_{i_1}\cdots e_{i_k},e_{j_2}\cdots e_{j_s}}+\cdots +\right.\\
&\qquad\left.\elbows{\sum _tB_{j_st}(e_t)e_{i_1}e_{i_2}\cdots e_{i_k},e_{j_1}e_{j_2}\cdots e_{j_{s-1}}}\right]\\
&=\elbows{B^\alpha e_{i_1}e_{i_2}\cdots e_{i_r},e_{j_1}e_{j_2}\cdots e_{j_s}}
\end{align*}
It follows that $B^\alpha$ is self-adjoint.
\end{proof}

\begin{example}  % Example 5
We now illustrate the proof of Theorem~\ref{thm61} with the example $H=\complex ^3$. Let
$\alpha =\brac{\alpha _1,\alpha _2,\alpha _3}\subseteq\real ^3$ and let $B\in\lscript _S(H)$ with $B(e_i)=\sum _jB_{ij}e_j$ so that 
$B_{ij}=\Bbar _{ji}$. Unlike the proof of Theorem~\ref{thm61}, we treat the various cases individually. Clearly,
$\elbows{e_1,B^\alpha e_2}=\elbows{B^\alpha e_1,e_2}$, $\elbows{e_1,B^\alpha e_1e_2}=\elbows{B^\alpha e_1,e_1e_2}=0$,
$\elbows{e_1,B^\alpha\iscript}=\elbows{B^\alpha e_1,\iscript}=0$, $\elbows{1,B^\alpha e_1}=\elbows{B^\alpha 1,e_1}=0$. We also have
\begin{align*}
\elbows{e_1e_2,B^\alpha\iscript}&
   =\alpha _3\sqbrac{\elbows{e_1e_2,Be_1e_2e_3}+\elbows{e_1e_2,e_1Be_2e_3}+\elbows{e_1e_2,e_1e_2Be_3}}\\
   &=0=\elbows{B^\alpha e_1e_2,\iscript}
\end{align*}
Moreover,
\begin{align*}
\elbows{e_1e_2,B^\alpha e_1e_2}&=\alpha _2\sqbrac{\elbows{e_1e_2,Be_1e_2}+\elbows{e_1e_2,e_1Be_2}}\\
&=\alpha _2\sqbrac{\elbows{e_1e_2,B_{11}e_1e_2+B_{12}e_2e_2+B_{13}e_3e_2}}\\
&\qquad +\alpha _2\sqbrac{\elbows{e_1e_2,e_1B_{21}e_1+e_1B_{22}e_2+e_1B_{23}e_3}}\\
&=\alpha _2(B_{11}+B_{22})=\alphabar _2(\Bbar _{11}+\Bbar _{22})=\overline{\elbows{e_1e_2,B^\alpha e_1e_2}}\\
&=\elbows{B^\alpha e_1e_2,e_1e_2}
\end{align*}
and finally
\begin{align*}
\elbows{e_1e_2,B^\alpha e_1e_3}&=\alpha _2\sqbrac{\elbows{e_1e_2,Be_1e_3}+\elbows{e_1e_2,e_1Be_3}}\\
&=\alpha _2\left[\elbows{e_1e_2,(B_{11}e_1+B_{12}e_2+B_{13}e_3)e_3}\right.\\
&\qquad +\left.\elbows{e_1e_2,e_1(B_{31}e_1+B_{32}e_2+B_{33}e_3)}\right]\\
&=\alpha _2B_{32}=\alpha _2\Bbar _{23}\\
&=\alpha _2\left[\elbows{(B_{11}e_1+B_{12}e_2+B_{13}e_3)e_2,e_1e_3}\right.\\
&\qquad\left. +\elbows{e_1(B_{21}e_1+B_{22}e_2+B_{23}e_3),e_1e_3}\right]\\
&=\alpha _2\sqbrac{\elbows{Be_1e_2,,e_1e_3}+\elbows{e_1Be_2,e_1e_3}}\\
&=\elbows{B^\alpha e_1e_2,e_1e_3}\hskip 14pc\square
\end{align*}
\end{example}

A great simplification occurs if $A\in\lscript (H)$ is diagonal with respect to the basis $e_1,e_2,\ldots ,e_n$. In this case
$A=\sum\limits _{i=1}^n\lambda _iP_i$ where $\lambda _i\in\real$ and $P_i$ is the projection onto $e_i$, $i=1,2,\ldots ,n$. If
$\alpha =(\alpha _1,\alpha _2,\ldots ,\alpha _n)\in\real ^n$ we obtain $A^\alpha\in\lscript _S(\gscript (H))$ with
$A^\alpha (1)=\alpha _1$, $A^\alpha (e_i)=A(e_i)=\lambda _ie_i$, $i=1,2,\ldots ,n$,
\begin{align*}
A^\alpha (e_ie_j)&=\alpha _2\sqbrac{A(e_i)e_j+e_iA(e_j)}=\alpha _2(\lambda _i+\lambda _j)e_ie_j\\
A^\alpha (e_ie_je_k)&=\alpha _3\sqbrac{A(e_i)e_je_k+e_iA(e_j)e_k+e_ie_jA(e_k)}\\
&=\alpha _3(\lambda _i+\lambda _j+\lambda _k)e_ie_je_k\\
&\quad\vdots\\
A^\alpha (\iscript )&=\alpha _n(\lambda _1+\lambda _2+\cdots +\lambda _n)e_1e_2\cdots e_n
\end{align*}
The eigenvalues of $A^\alpha$ are $\alpha _1,\lambda _i$, $i=1,2,\ldots ,n$, $\alpha _2(\lambda _i+\lambda _j)$, $i,j=1,2,\ldots ,n$,
$\alpha _3(\lambda _i+\lambda _j+\lambda _k)$, $i,j,k=1,2,\ldots ,n\ldots ,\alpha _n(\lambda _1+\lambda _2+\cdots +\lambda _n)$. The corresponding eigenvectors are the basis $1,e_i,e_ie_j,e_ie_je_k,\ldots ,\iscript$. Considering $A^\alpha$ to be the Hamiltonian for the system, the corresponding dynamics is given by
\begin{align*}
U_t^\alpha (1)&=e^{i\pi\alpha _1t}1\\
U_t^\alpha (e_j)&=e^{i\pi\lambda _jt}e_j\\
U_t^\alpha (e_re_j)&=e^{i\pi\alpha _2(\lambda _r+\lambda _h)t}e_re_j\\
U_t^\alpha (e_re_je_k)&=e^{i\pi\alpha _3(\lambda _r+\lambda _j+\lambda _k)t}e_re_je_k,\ldots\\
U_t^\alpha (\iscript )&=e^{i\pi\alpha _n(\lambda _1+\lambda _2+\cdots +\lambda _n)t}\iscript
\end{align*}

We close by showing that this work extends to infinite dimensional separable Hilbert spaces.

\begin{thm}    % Theorem 6.2
\label{thm62}
Let $H$ be a separable infinite dimensional Hilbert space with orthonormal basis $e_1,e_2,\ldots\,$. Then there exists a unique separable infinite dimensional Hilbert geometric algebra $\gscript (H)$ with the following properties:
{\rm{(i)}}\enspace $H\subseteq\gscript (H)$,
{\rm{(ii)}}\enspace If $u\in H$, then $\elbows{\utilde,u}=uu$.
{\rm{(iii)}}\enspace $\gscript (H)$ has the orthonormal bassis given by
\begin{align*}
&1\\
&\brac{e_i\colon i=1,2,\ldots}\\
&\brac{e_ie_j\colon i<j,\ i,j=1,2,\ldots}\\
&\brac{e_ie_je_k\colon i<j<k,\ i,j,k=1,2,\ldots}\\
&\brac{e_ie_je_k\colon i<j<k,\ i,j,k=1,2,\ldots}\\
&\qquad\vdots\\
&\brac{e_{i_1}e_{i_2}\cdots e_{i_n}\colon i_1<i_2\cdots <i_n,\ i_1,i_2,\ldots ,i_n=1,2,\ldots}\\
&\qquad\vdots
\end{align*}
\end{thm}
\begin{proof}
Let $H_n$ be the $n$-dimensional subspace generated by $e_1,e_2,\ldots ,e_n$, $n=1,2,\ldots\,$. Then the $2^n$-dimensional Hilbert geometric algebra $\gscript (H_n)$ exists \cite{mac17} and $\gscript (H_n)\subseteq\gscript (H_{n+1})$, $n=1,2,\ldots\,$. Let
$\gscript _0(H)=\bigcup\limits _{n=1}^\infty\gscript (H_n)$. For $a,b\in\gscript _0(H)$, we have $a,b\in\gscript (H_n)$ for some $n$ and we define $\elbows{a,b}=\elbows{a,b}_n$ in $\gscript (H_n)$ in which case $\elbows{a,b}$ does not depend on $n$. It is clear that 
$\elbows{\tbullet,\tbullet}$ is an inner product so $\gscript _0(H)$ is an inner product space with orthonormal basis given by the elements listed in (iii). If $\gscript (H)$ is the completion of $\gscript _0(H)$, then $\gscript (H)$ is the smallest Hilbert space containing
$\gscript _0(H)$. It follows that the listed elements in (iii) form an orthonormal basis for $\gscript (H)$. A sequence $a_i\in\gscript (H)$ is
\textit{Cauchy} if for any $\epsilon >0$ there exists an integer $N_\epsilon$ such that $i,j\ge N_\epsilon$ implies
$\doubleab{a_i-a_j}<\epsilon$. We then have that $a\in\gscript (H)$ if and only if there exists a Cauchy sequence $a_i\in\gscript _0(H)$ such that $\lim\limits _{i\to\infty}\doubleab{a_i-a}=0$ so $\lim\limits _{i\to\infty}a_i=a$. To verify (i), letting $a\in H$ we have
$a=\sum\limits _{i=1}^\infty c_ie_i$, $c_i\in\complex$. Then we have $a=\lim a_n=\sum\limits _{i=1}^nc_ie_i$ where
$a_n\in H_n\subseteq\gscript _0(H)$. Hence, $a\in\gscript (H)$ so (i) holds. We now show that $\gscript (H)$ is a geometric algebra. if
$a,b\in\gscript (H)$ then there exist $a_n,b_n\in\gscript _0(H)$ such that $\lim a_n=a,\lim b_n=b$ and we can assume that
$a_n,b_n\in\gscript (H_n)$. Letting $c_n=a_nb_n$ we have that $c_n\in\gscript (H_n)$ and
\begin{align*}
\doubleab{c_n-c_m}&=\doubleab{a_nb_n-a_mb_m}\le\doubleab{a_nb_n-a_nb_m}+\doubleab{a_nb_m-a_mb_m}\\
   &=\doubleab{a_n(b_n-b_m)}+\doubleab{(a_n-a_m)b_m}
\end{align*}
We can consider $c\to a_nc$ as a linear operator on $\gscript (H_n)$. Since $\gscript (H_n)$ is finite dimensional, this operator is bounded with norm $\doubleab{a_n}$. Since $\lim a_n=a$ there exists a $K\in\real ^+$ such that $\doubleab{a_n}\le K$ for every $n$ and similarly $\doubleab{b_m}\le M$ for every $m$, Hence,
\begin{equation*}
\doubleab{c_n-c_m}\le K\doubleab{b_n-b_m}+M\doubleab{a_n-a_m}
\end{equation*}
Therefore, $c_n$ is a Cauchy sequence and we define the product on $\gscript (H)$ by
\begin{equation*}
a\tbullet b=\lim c_n=\lim a_nb_n
\end{equation*}
It follows that if $a,b\in\gscript _0(H)$, then $a,b\in\gscript (H_n)$ for some $n$ and $a\tbullet b=ab$ so the product $a\tbullet b$ extends that on $\gscript _0(H)$. To verify (ii), suppose $u\in H$. Then there exist $u_n\in H_n\subseteq\gscript (H_n)\subseteq\gscript _0(H)$ with
$\lim u_n=u$. Then 
\begin{equation*}
u\tbullet u=\lim u_nu_n=\lim\elbows{\utilde _n,u_n}=\elbows{\utilde ,u}
\end{equation*}
so (ii) holds. To show that $\gscript (H)$ is a geometric algebra, it is clear that $a\tbullet b$ is homogeneous. To show associativity, if
$a,b,c\in\gscript (H)$, there exists $a_n,b_n,c_n\in\gscript _0(H)$ such that $\lim a_n=a$, $\lim b_n=b$ and $\lim c_n=c$.
We then have
\begin{align*}
a\tbullet (b\tbullet c)&=\lim a_n\tbullet (\lim b_n\tbullet\lim c_n)=\lim a_n\tbullet\sqbrac{lim (b_nc_n)}\\
   &=\lim a_nb_nc_n=\lim a_nb_n\tbullet\lim c_n=a\tbullet b\tbullet c_n=(a\tbullet b)\tbullet c
\end{align*}
To show distributivity, we have
\begin{align*}
a\tbullet (b+c)&=\lim a_n\tbullet\sqbrac{\lim (b_n+c_n)}=\lim a_n(b_n+c_n)\\
   &=\lim a_nb_n+\lim a_nc_n=a\tbullet b+a\tbullet c
\end{align*}
It follows that $\gscript (H)$ is a geometric algebra satisfying (i), (ii) and (iii). The uniqueness of $\gscript (H)$ is clear. Finally, assuming the axiom of choice, it follows that a countable union of countable sets is countable. We conclude that the orthonormal basis listed in (iii) is countable so $\gscript (H)$ is separable.
\end{proof}


\begin{thebibliography}{99}
% ref 1
\bibitem{art11}J.\,Arthur, \textit{Understanding Geometric Algebra for Electromagnetic Theory}, Wiley-IEEE Press, 2011.
% ref 2
\bibitem{bgl95}P.\,Busch, M.\,Grabowski and P.\,Lahti, \textit{Operational Quantum Physics}, Springer-Verlag, Berlin, 1997.
% ref 3
\bibitem{blm96}P.\,Busch, P.\,Lahti, and P.\,Mittlestaedt, \textit{The Quantum Theory of Measurement}, Springer-Verlag, Berlin, 1996.
% ref 4
\bibitem{dl03}C.\,Doran and A.\,Lasenby, \textit{Geometric Algebra for Physicists}, Cambridge University Press, Cambridge, 2003.
% ref 5
\bibitem{dfm09}L.\,Dorst, D.\,Fontijne and S.\,Mann, \textit{Geometric Algebra for Computer Science}, Morgan Kaufman, Second Printing, 2009.
% ref 6
\bibitem{dor02}L.\,Dorst, in L.\,Dorst, C.\,Doren and J.\,Lasenly (Eds.) Applications in Computer and Engineering, Birkhauser, Boston, 2002 (35--46)
% ref 7
\bibitem{dys77}F.\,Dyson, Missed Opportunities, \textit{Bull.\,Am.\,Math.\,Soc.} \textbf{78}, 635--652 (1977).
% ref 8
\bibitem{hz12}T.\,Heinosaari and M.\,Ziman, \textit{The Mathematical Language of Quantum Theory}, Cambridge University Press, Cambridge, 2012.
% ref 9
\bibitem{hs84}D.\,Hestenes and G.\,Sobezyk, \textit{Clifford Algebra to Geometric Calculus}, Academic Publishers, 1984.
% ref 10
\bibitem{hs03}D.\,Hestenes, Reforming the mathematics of physics, \textit{Am.\,J.\,Phys.},  \textbf{71}, 104--121 (2003).
% ref 11
\bibitem{hs15}---, \textit{Space-Time Algebra}, 2nd Edition, Springer-Verlag, Berlin 2015.
% ref 12
\bibitem{kra83}K.\,Kraus, States, Effects and Operations, \textit{Lecture Notes in Physics}, \textbf{190}, Springer, Berlin, 1983.
% ref 13
\bibitem{mac17}A.\,Macdonald, Geometric algebra and geometric calculus, \textit{Adv.\ Appl.\ Cliff.\ Alg.}, \textbf{27}, 853--891 (2017).
% ref 14
\bibitem{nc00}M.\,Nielson and I.\,Chuang, \textit{Quantum Computation and Quantum Information}, Cambridge University Press, Cambridge, 2000.

\end{thebibliography}
\end{document}